\documentclass[11pt]{article}

\usepackage{amsmath,amssymb,amstext,amsthm,enumerate}
\usepackage[usenames,dvipsnames]{xcolor}
\usepackage{tikz}

\usepackage[top = 1in, bottom = 1in, left = 1in, right = 1in]{geometry}

\newcommand{\p}[1]{\mathbb{P}\left[{#1}\right]}
\newcommand{\e}[1]{\mathbb{E}\left[{#1}\right]}
\newcommand{\eq}[1]{(\ref{#1})}
\newcommand{\eps}{\varepsilon}
\newcommand{\Oh}{{O}}
\newcommand{\diam}{\operatorname{diam}}
\newcommand{\sg}{G_n^*}
\newcommand{\pg}{P_n}

\newcommand{\erv}{\operatorname{Exp}}
\newcommand{\sts}{\operatorname{ST}_{\mathsf s}}
\newcommand{\sta}{\operatorname{ST}_{\mathsf a}}
\newcommand{\wastspp}{\operatorname{ast}_{\mathsf s}}
\newcommand{\wastapp}{\operatorname{ast}_{\mathsf a}}
\newcommand{\wastapull}{\operatorname{ast}_{\mathsf a}^{\mathsf {pull}}}
\newcommand{\gstspp}{\operatorname{gst}_{\mathsf s}}
\newcommand{\gstapp}{\operatorname{gst}_{\mathsf a}}
\def\1{{\mathchoice {1\mskip-4mu\mathrm l}      
{1\mskip-4mu\mathrm l} 
{1\mskip-4.5mu\mathrm l} {1\mskip-5mu\mathrm l}}} 
 
\def\comment#1{} 
\newcommand{\samedist}{{\,\buildrel d \over =\,}}
\newcommand{\stochle}{{\,\stackrel{s}{\le}\,}}

\theoremstyle{plain}
\newtheorem{theorem}{Theorem}
\newtheorem{lemma}[theorem]{Lemma}
\newtheorem{proposition}[theorem]{Proposition}
\newtheorem{corollary}[theorem]{Corollary}

\theoremstyle{definition}
\newtheorem*{definition}{Definition}
\newtheorem{conjecture}[theorem]{Conjecture}
\newtheorem{open}[theorem]{Open problem}

\newtheorem{observation}[theorem]{Observation}

\theoremstyle{remark}

\newcommand{\geo}{\operatorname{Geo}}
\usetikzlibrary{shapes.geometric} 
\tikzstyle{vtx}=[circle, draw, fill=black!100, inner sep=0pt, minimum size=4pt]
\tikzstyle{evtx}=[circle, draw, inner sep=0pt, minimum size=4pt]
\newcommand{\Nz}{{\mathbb{N}_0}}
\newcommand{\N}{{\mathbb{N}}}
\newcommand{\cond} {{\,\left\vert\vphantom{\frac{1}{1}}\right.\,}}

\title{On the push\&pull protocol for rumour spreading\thanks{A preliminary version of this paper has appeared in proceedings of the ACM Symposium on Principles of Distributed Computing (PODC 2015), pages 405--412.}}

\author{
H\"{u}seyin Acan\\
{\small School of Mathematical Sciences, Monash University} \\
{\small {\tt  huseyin.acan@monash.edu}} \and
Andrea Collevecchio\thanks{Supported by ARC Discovery Project grant DP140100559 {and ERC STREP project MATHEMACS}.}\\
{\small School of Mathematical Sciences, Monash University, and}\\
{\small Ca' Foscari University, Venice}\\
{\small {\tt  andrea.collevecchio@monash.edu }} 
\and
Abbas  Mehrabian\thanks{Supported by the Vanier Canada Graduate Scholarships program.}\\
{\small Department of Combinatorics and Optimization, University of Waterloo}\\
{\small {\tt amehrabi@uwaterloo.ca}}
\and
Nick Wormald\thanks{Supported by Australian Laureate Fellowships grant FL120100125.}\\
{\small School of Mathematical Sciences, Monash University} \\
{\small {\tt  nick.wormald@monash.edu }} }
\date{}

\begin{document}

\maketitle

\vspace{-1cm}
\begin{abstract}
The asynchronous push\&pull protocol, a randomized distributed algorithm
for spreading a rumour in a graph $G$, is defined as follows.
Independent exponential clocks of rate 1 are associated with the vertices of
$G$, one to each vertex.
Initially, one vertex of $G$ knows the rumour.
Whenever the clock of a vertex $x$ rings, it calls a random neighbour $y$:
if $x$ knows the rumour and $y$ does not, then $x$ tells $y$ the rumour (a
push operation),
and if $x$ does not know the rumour and $y$ knows it,
$y$ tells $x$ the rumour (a pull operation).
The average spread time of $G$ is the expected time it takes for all
vertices to know the rumour, and the guaranteed spread time of $G$ is the
smallest time $t$ such that with probability at least $1 - 1/n$, after time
$t$ all vertices know the rumour.
The synchronous variant of this protocol, in which each clock rings
precisely at times $1,2,\dots$,
has been studied extensively.

We prove the following results for any $n$-vertex graph:
In either version, the average spread time is at most linear even if only
the pull operation is used,
and the guaranteed spread time is within a logarithmic factor of the
average spread time, so it is $O(n \log n)$.
In the asynchronous version, both the average and guaranteed spread times
are $\Omega(\log n)$.  We give examples of graphs illustrating that
these bounds are best possible up to constant factors.

We also prove the first analytical relationships between the guaranteed
spread times in the two versions. Firstly, in all graphs the guaranteed spread time in the asynchronous
version is within an $O(\log n)$ factor of that in the synchronous version,
and this is tight. Next, we find examples of graphs whose asynchronous spread times are
logarithmic, but the synchronous versions are polynomially large. 
Finally, we show for any graph that the ratio of the 
synchronous spread time to the asynchronous spread time is
$O\big(n^{2/3}\big)$. 
\end{abstract}

\section{Introduction}
\label{sec:intro}
Randomized rumour spreading is an important primitive for information dissemination in networks and has
numerous applications in network science, ranging from spreading information in the WWW and Twitter to spreading viruses and diffusion of ideas in human communities. A well studied rumour spreading protocol is the  \emph{(synchronous) {push\&pull} protocol},
introduced by Demers, Greene, Hauser, Irish, Larson, Shenker, Sturgis, Swinehart, and Terry~\cite{DGH+87}
and popularized by Karp, Schindelhauer, Shenker, and V\"ocking~\cite{KSSV00}.
Suppose that one node in a  network is aware of a piece of information, the `rumour',
and wants to spread it to all nodes quickly.
The protocol proceeds in rounds.
In each round, every \emph{informed} node contacts  a random neighbour and sends the rumour to it (`pushes' the rumour),
and every \emph{uninformed} nodes contacts a random neighbour and gets the rumour if the neighbour  knows it (`pulls' the rumour).

A point to point communication network can be modelled as an undirected graph: the nodes represent the processors and the links represent communication channels between them. 
Studying rumour spreading has several applications to distributed computing in such networks, of which we mention just two. 
The first is in broadcasting algorithms: a single processor wants to broadcast a piece of information to all other processors in the network (see~\cite{broadcasting_survey} for a survey). 
There are at least four advantages to 
the push\&pull protocol:
it puts much less load on the edges than naive flooding,
it is simple (each node makes a simple local decision in each round; no knowledge of the global topology is needed; no state is maintained), scalable (the protocol is independent of the size of network: it does not grow more complex as the network grows) and robust (the protocol tolerates random node/link failures without the use of error recovery mechanisms, see~\cite{FPRU90}).
A second application comes from the maintenance of
databases replicated at many sites, e.g., yellow pages, name servers, or server directories. 
There are updates injected at various nodes, and these updates must propagate to all nodes in the network. 
In each round, a processor communicates with a random neighbour and they share any new information, so that eventually all copies of the database converge to the same contents. See~\cite{DGH+87} for details.
Other than the aforementioned applications, rumour spreading protocols have successfully been applied in various contexts such as resource discovery~\cite{Harchol-Balter1999},
distributed averaging~\cite{Boyd2006}, data aggregation~\cite{KDG03},
and the spread of computer viruses~\cite{BBCS05}.

In this paper we only consider simple, undirected and connected graphs.
Given a graph and a starting vertex, the \emph{spread time} of a certain protocol is the time it takes for
the rumour to spread in the whole graph, i.e.\ the time difference between the moment the protocol is initiated and the moment when everyone learns the rumour.
For the synchronous push\&pull protocol, it turned out that the spread  time is closely related to the \emph{expansion profile} of the graph.
Let $\Phi(G)$ and $\alpha(G)$ denote the conductance
and the vertex expansion
of a graph $G$, respectively.
After a series of results by various scholars, Giakkoupis~\cite{Gia11,G13} showed the spread time  is
${\Oh}\left(\min\{\Phi(G)^{-1}\cdot{\log n}, \alpha(G)^{-1} \cdot \log^2 n \}\right)$.
This protocol has recently been used to model news propagation in social networks.
Doerr, Fouz, and Friedrich~\cite{DFF11} 
proved an upper bound of $O(\log n)$ for the spread time on Barab\'{a}si-Albert graphs, and
Fountoulakis, Panagiotou, and Sauerwald~\cite{FPS12} proved the same upper bound (up to constant factors) for the spread time on Chung-Lu random graphs.

All the above results assumed a synchronized model, i.e.\  all nodes take action simultaneously at discrete time steps. 
In many applications and certainly in real-world social networks, this assumption is not very plausible. 
Boyd, Ghosh, Prabhakar, Shah~\cite{Boyd2006} proposed an asynchronous time model with a continuous time line. 
Each node has its own independent clock that rings at the times of a rate 1 Poisson process.
(Since the times between rings is an exponential random variable, we shall call this an \emph{exponential clock}.)
The protocol now specifies for every node what to do when its own clock rings.
The rumour spreading problem in the asynchronous time model has so far received less attention. 
{Rumour spreading protocols in this model turn out
to be closely related to Richardson's model for the spread of a disease~\cite{richardson_model_survey} and to first-passage percolation~\cite{fpp_book}  with edges having i.i.d.\ exponential weights. 
The main difference is that in rumour spreading protocols each vertex contacts one neighbour at a time. 
So, for instance in the `push only' protocol, the net communication rate outwards from a vertex is fixed, and hence the rate that the vertex passes the rumour to any one given neighbour is inversely proportional to its degree (the push\&pull protocol is a bit more complicated). 
Hence, the degrees of vertices play a crucial role not seen in Richardson's model or first-passage percolation.
However, on  regular graphs, the asynchronous push\&pull protocol,
Richardson's model, and first-passage percolation are essentially the same process, assuming appropriate parameters are chosen.}
In this sense, Fill and Pemantle~\cite{pemantle}
and Bollob{\'a}s and Kohayakawa~\cite{bol} showed that
a.a.s.\ the spread time {of the asynchronous push\&pull protocol} is $\Theta(\log n)$ on the hypercube graph. 
Janson~\cite{asynchronous_complete} {and Amini, Draief and Lelarge~\cite{dregular_asynch}} showed the same results (up to constant factors) for the complete graph
{and for random regular graphs, respectively.}
These bounds match the same
order of magnitude as in the synchronized case. 
Doerr, Fouz, and Friedrich~\cite{experimental} experimentally compared the spread time in the two time models.
They state that
`Our experiments show that the asynchronous model
is faster on all graph classes [considered here].'
However, a general relationship between the spread times of the two variants has not been proved theoretically.

Fountoulakis, Panagiotou, and Sauerwald~\cite{FPS12} studied the asynchronous push\&pull protocol on  Chung-Lu random graphs with exponent between $2$ and $3$. 
For these graphs, they showed that a.a.s.\ after some constant time, $n-o(n)$ nodes are informed.
Doerr, Fouz, and Friedrich~\cite{DFF12} showed that for the preferential attachment graph (the non-tree case), a.a.s.\ all but $o(n)$ vertices receive the rumour in time $O\left(\sqrt {\log n} \right)$, but to inform
all vertices a.a.s., $\Theta(\log n)$ time is necessary and sufficient.
Panagiotou and Speidel~\cite{speidel} studied this protocol on Erd\H{o}s-Renyi random graphs and proved that if the average degree is $(1+\Omega(1))\log n$, a.a.s.\ the spread time is $(1+o(1))\log n$.

\subsection{Our contribution}
{
In this paper we answer a  fundamental question about the asynchronous push\&pull protocol: what are the minimum and maximum spread times on an $n$-vertex graph?
Our proof techniques yield new results on the well studied synchronous version as well.
We also compare the performances of the two protocols on the same graph, and prove the first theoretical relationships between their spread times.

We now formally define the protocols.
In this paper $G$ denotes the ground graph which is simple and connected, and 
$n$ counts its vertices, and is assumed to be sufficiently large.

\begin{definition}[Asynchronous push\&pull protocol]
Suppose that an independent exponential clock of rate 1 is associated with each vertex of $G$.
Suppose that initially, some vertex $v$ of $G$ knows a piece of information, the so-called \emph{rumour}. 
The rumour spreads in $G$ as follows.
Whenever the clock of a vertex $x$ rings,
this vertex performs an `action':
it calls a random neighbour $y$;
if $x$ knows the rumour and $y$ does not, then $x$ tells $y$ the rumour (a \emph{push} operation),
and if $x$ does not know the rumour and $y$ knows it,
$y$ tells $x$ the rumour (a \emph{pull} operation).
Note that if both $x$ and $y$ know the rumour or neither of them knows it, then this action is useless.
Also, vertices have no memory, hence $x$ may call the same neighbour several consecutive times.
The \emph{spread time} of $G$ starting from $v$, written $\sta(G,v)$, is the first time that all vertices of $G$ know the rumour.
Note that this is a continuous random variable, with two sources of randomness: the Poisson processes associated with the vertices, and random neighbour-selection of the vertices.
The \emph{guaranteed spread time} of $G$, written $\gstapp(G)$, is the smallest deterministic number $t$ such that for every $v\in V(G)$ we have 
$\p{\sta(G,v) > t} \le 1/n$. 
The \emph{average spread time} of $G$, written $\wastapp(G)$, is the smallest deterministic number $t$ such that for every $v\in V(G)$ we have 
$\e{\sta(G,v)} \le t$. 
\end{definition}

\begin{definition}[Synchronous push\&pull protocol]
Initially some vertex $v$ of $G$ knows the rumour, which spreads in $G$ in a round-robin manner:
in each round $1,2,\dots$, all vertices perform actions simultaneously. 
That is, each vertex $x$ calls a random neighbour $y$;
if $x$ knows the rumour and $y$ does not, then $x$ tells $y$ the rumour (a \emph{push} operation),
and if $x$ does not know the rumour and $y$ knows it,
$y$ tells $x$ the rumour (a \emph{pull} operation).
Note that this is a synchronous protocol, e.g.\ a vertex that receives a rumour in a certain round cannot send it on in the same round.
The \emph{spread time} of $G$ starting from $v$,  $\sts(G,v)$, is the first time that all vertices of $G$ know the rumour.
Note that this is a discrete random variable, with one source of randomness: the random neighbour-selection of the vertices.
The \emph{guaranteed spread time} of $G$, written $\gstspp(G)$, and
the \emph{average spread time} of $G$, written $\wastspp(G)$, 
are defined in an analogous way to the asynchronous case.
\end{definition}

We remark that the notion of `guaranteed spread time' was first defined by Feige, Peleg, Raghavan and Upfal~\cite{FPRU90} under the name `almost sure rumor coverage time' for the `push only' protocol.
(In this protocol, which was studied prior to push\&pull, the informed nodes push the rumour, but the uninformed ones do nothing.
The `pull only' protocol is defined conversely.)

It turns out that changing the starting vertex affects the spread time by at most a multiplicative factor of 2.
Specifically, in Proposition~\ref{pro:start}
we prove that for any two vertices $u$ and $v$,
$\sts(G,u) \stochle 2 \sts(G,v)$
and
$\sta(G,u) \stochle 2 \sta(G,v)$.
(For random variables $X$ and $Y$, $X \stochle Y$ means $X$ is \emph{stochastically dominated by} $Y$, that is, for any $t$,
$\p{X\ge t} \le \p{Y \ge t}$.)
These imply $\wastspp(G) \le 2 \e{\sts(G,v)}$
and $\wastapp(G) \le 2 \e{\sta(G,v)}$
for any vertex $v$. 

Our first main result is the following theorem.

\begin{theorem}
\label{thm:extremal}
For large enough $n$, the following hold  for any $n$-vertex graph $G$.
{\begin{align}
(1-1/n) \wastapp(G) & \le \gstapp(G) \le 
e  \wastapp(G) \log n \:,
\label{expectedguaranteed}\\
\frac15 \log n & <  \wastapp (G)  < 4 n \:,  
\label{expectedextremal}\\
  \frac15 \log n & \leq \gstapp (G) \leq 4 e n \log n \:.\label{guaranteedextremal}
\end{align}}
Moreover, these bounds are {asymptotically  best possible, up to the constant factors}.
\end{theorem}
\noindent Our proof of the right-hand bound in~\eq{expectedextremal} is based on the pull operation only, so this bound applies equally well to the `pull only' protocol. 

The arguments
for~\eq{expectedguaranteed} and the right-hand bounds in~\eq{expectedextremal} and~\eq{guaranteedextremal} can easily be extended to the synchronous variant, giving the following theorem. 
{The bound~\eq{guaranteedextremals} below also follows from~\cite[Theorem~4]{robustness}, but here we also show its tightness.}

\begin{theorem}
\label{thm:extremalsynchronous}
The following   hold for any $n$-vertex graph $G$.
\begin{align}
(1-1/n) \wastspp(G)  \le \ \gstspp(G) \ & \le 
e  \wastspp(G) \log n \:,
\label{expectedguaranteeds}\\
 \wastspp (G)  & < 4.6 n
\label{expectedextremals}\:,\\
\gstspp (G) & < 4.6 e  n \log n \:.\label{guaranteedextremals}
\end{align}
Moreover, these bounds are {asymptotically  best possible, up to the constant factors}.
\end{theorem}

\begin{open}
\label{openbestconstants}
Find the best possible constants factors in Theorems~\ref{thm:extremal}
and~\ref{thm:extremalsynchronous}.
\end{open}

We next turn to studying the relationship between the asynchronous and synchronous variants on the same graph.
Let ${H}_n := 
\sum_{i=1}^{n} 1/i $ denote the $n$th harmonic number.
It is well known that ${H}_n=\log n + O(1)$.

\begin{theorem}
\label{thm:allring}
For any $G$ we have
$\wastapp (G) \leq  {H}_n \times  \wastspp(G) $
and 
$\gstapp (G) \leq 8  \gstspp(G) \log n $,
and these bounds are best possible, up to the constant factors.
\end{theorem}

For all graphs we examined stronger results hold, which suggests the following conjecture.

\begin{conjecture}
For any $n$-vertex graph $G$ we have
$\wastapp(G) \le \wastspp(G) + O(\log n)$
and
$\gstapp(G) \le \gstspp(G) + O(\log n)$.
\end{conjecture}
Our last main result is the following theorem, whose proof is somewhat technical, and uses couplings with the sequential rumour spreading protocol.

\begin{theorem}
\label{thm:last}
For any $\alpha \in [0,1)$ we have
$${\gstspp(G)} \le  n^{1-\alpha} + 64 {\gstapp(G)} n ^{(1+\alpha)/2} \:.$$
\end{theorem}
\begin{corollary}
\label{cor:ratio}
We have 
$$  
\frac{\gstspp(G)}{\gstapp(G)} = \Omega(1 / \log n)
\quad
\mathrm{and}
\quad
\frac{\gstspp(G)}{\gstapp(G)}
= O\big(n^{2/3}\big)\:,$$
and the left-hand {bound is asymptotically} best possible, up to the constant factor. 
 Moreover, there exist {infinitely many} graphs for which this ratio is
$\Omega \left(n^{1/3}(\log n)^{-4/3} \right)$.
\end{corollary}

\begin{open}
\label{open}
What is the maximum possible value
of the ratio $\gstspp(G)/\gstapp(G)$
for an $n$-vertex graph $G$?
\end{open}

We make the following conjecture.

\begin{conjecture}
\label{conj}
For any $n$-vertex graph $G$ we have 
$$ \frac{\gstspp(G)}{\gstapp(G)}
= O\left(\sqrt n\, (\log n)^{O(1)}\right) \:,$$
and this is tight for infinitely many graphs.
\end{conjecture} 
 
The parameters $\wastspp(G)$ and $\wastapp(G)$ can be approximated easily using the Monte Carlo method: simulate the protocols several times, measuring the spread time of each simulation, and output the average.
Another open problem is to design a \emph{deterministic} approximation algorithm for any one of 
$\wastapp(G)$, $\gstapp(G)$, $\wastspp(G)$ or $\gstspp(G)$.

For the proofs we use standard graph theoretic arguments and well known properties of the exponential distribution and Poisson processes,
in particular the memorylessness, and the fact that the union of two Poisson processes is another Poisson process.
For proving Theorem~\ref{thm:last} we define a careful coupling between the synchronous and asynchronous protocols.

Previous work on the asynchronous push\&pull protocol has focused on special graphs.
This paper is the first systematic study of this protocol on all graphs.
We believe this protocol is fascinating and is quite different from its synchronous variant, in the sense that different techniques are required for analyzing it, and the spread times of the two versions can be quite different. 
Our work makes significant progress on better understanding of this protocol,
and will hopefully inspire further research on this problem.

A collection of known and new bounds for the average spread times of many graph classes is given in Table~\ref{table}.
In Section~\ref{sec:prelim} we prove some preliminary results and study some examples, which demonstrate tightness of some of the above bounds.
Theorems~\ref{thm:extremal}
and~\ref{thm:extremalsynchronous} are proved in Section~\ref{sec:thm:extremal}.
Theorems~\ref{thm:allring} and~\ref{thm:last} and
Corollary~\ref{cor:ratio}
are proved in Section~\ref{sec:cor:ratio}.

\begin{table}
\begin{center}
\renewcommand{\arraystretch}{1.2}
\renewcommand{\tabcolsep}{0.1cm}
\begin{tabular}{c || c | c}
Graph $G$ & $\wastspp(G)$ & $\wastapp(G)$ \\ \hline\hline
Path & $(4/3)n + O(1)$ & $n + O(1)$ \\  \hline

Star & $2$ & $\log n + O(1)$\\ \hline

Complete & $\sim \log_3 n$ & $\log n + o(1)$\\
&  \cite{KSSV00} & \cite{asynchronous_complete} \\ \hline

String of diamonds $\mathcal{S}_{m,k}$ & $\boldsymbol{\Omega(m)}$& $\boldsymbol{O(\log n+m/\sqrt k)}$\\ 
(see Section~\ref{sec:chain}) & & \\
\hline

Hypercube & $\Theta(\log n)$& $\Theta(\log n)$\\
& \cite{FPRU90} & \cite{pemantle}\\ \hline

Random graphs $\mathcal{G}(n,p)$ & $\Theta(\log n)$ & $\sim \log n$ \\
$1<\frac{np}{\log n}$ fixed & \cite{FPRU90} & \cite{speidel}\\\hline

Random $d$-regular graphs & $\Theta(\log n)$ &  $\sim (\log n)(d-1)/(d-2)$\\ 
$2 < d$ fixed & \cite{FP10}& \cite{dregular_asynch}\\ \hline

Preferential attachment graphs &$\Theta(\log n)$&$\Theta(\log n)$\\ 
(Barab\'{a}si-Albert model)&\cite{DFF11}&\cite{DFF12}\\ \hline

{\small Random geometric graphs in} & $\Theta(\sqrt[d]n/r + \log n)$ & $\boldsymbol{O\left({\log n}{\cdot}{\sqrt[d]n/r}+{\log^2 n}\right)}$ \\
{\small $\left[0,\sqrt[d]n\right]^d$ with edge threshold $r$}& \cite{rgg}& \\
{\small  above giant component threshold}&&\\\hline

Random $k$-trees &$\Omega\left(n^{1/(k+3)}\right)$&$\Omega\left(n^{1/(k+3)}\right)$\\ 
($2\leq k$ fixed)&\cite{we_ktrees}&\cite{we_ktrees}\\ \hline

General&
$O\left(\Delta(G)(\diam(G)+\log n)\right)$ & $O\left(\Delta(G)(\diam(G)+\log n)\right)$\\
& \cite{FPRU90} & \cite{FPRU90} \\ \hline

General&
$O\left((\log n) / \Phi(G)\right)$ & $\boldsymbol{O\left((\log^2 n) / \Phi(G)\right)}$\\
& \cite{Gia11} &  \\ \hline

General&
$O\left( (\log \Delta(G) \cdot \log n) / \alpha(G) \right)$ & $\boldsymbol{O\left((\log \Delta(G) \cdot \log^2 n) / \alpha(G)\right)}$\\
& \cite{G13} &  
\end{tabular}
\caption{Average spread times of some graph classes are shown.
For many of the entries, the relevant paper indeed proves an asymptotically almost sure bound for the spread time.
Bold entries are new to this paper.
Results for random graph classes hold \emph{asymptotically almost surely} as the number of vertices grows.
The notation $\sim$ means equality up to a $1+o(1)$ factor.
$\Delta(G)$ denotes the maximum degree of $G$.
For $S\subseteq V(G)$, let $\partial S$ be  the set of vertices in $V(G)\setminus S$ that have a neighbour in $S$, and let 
${ e}(S, V(G)\setminus S)$ be the number of edges between $S$ and $V(G)\setminus S$,
and let ${\tt vol}(S)=\sum_{u\in S}\deg(u)$.
Then $\alpha(G) :=\min\left\{\frac{|\partial S|}{|S|} :  S\subseteq V(G), 0<|S|\leq |V(G)|/2  \right\}$ and
$\Phi(G) :=\min\left\{\frac{e(S, V(G)\setminus S)}{{\tt vol}(S)} : S\subseteq V(G),
0 < {\tt vol}(S) \le {\tt vol}(V(G)) / 2 \right\}$.}
\label{table}
\end{center}
\end{table}

\section{Preliminaries and examples}
\label{sec:prelim}
Let us denote $\Nz=\{0,1,2,\dots\}$ and $\N=\{1,2,\dots\}$.
Let $\geo(p)$ denote a geometric random variable with parameter $p$ taking values in $\Nz$; namely for every  $k\in\Nz$,
$\p{\geo(p) = k} = (1-p)^kp$.
Let $\erv(\lambda)$ denote an exponential random variable with parameter $\lambda$ and mean $1/\lambda$.
For random variables $X$ and $Y$, $X\samedist Y$ means $X$ and $Y$ have the same distribution.
All logarithms are natural.
For functions $f$ and $g$,
$f \sim g$ means $\lim f(n)/g(n) = 1$ as $n$ grows.

We start by making a few observations valid of all graphs.

\begin{observation}
\label{obs:thinning}
Consider the asynchronous variant.
Let $uv$ be an edge.
Whenever $v$'s clock rings, it calls $u$ with probability $1/\deg(v)$.
Hence, for each vertex $v$, we can replace $v$'s clock by one exponential clock for each incident edge, these clocks being independent of all other clocks and having rate $1/\deg(v)$.
\end{observation}

\begin{observation}
\label{obs:restart}
Whenever a new vertex is informed, by memorylessness of the exponential random variable, we may imagine that all clocks are restarted.
\end{observation}

The following definition will be used throughout.

\begin{definition}[Communication time]
For an edge $e=uv$, the \emph{communication time via edge $e$}, written $T(e)$, is defined as follows.
Suppose $\tau$ is the first time that one of $u$ and $v$ learns the rumour, 
and $\rho$ is the first time after $\tau$ that one of $u$ and $v$ calls the other one. 
Then $T(e) = \rho - \tau$, which is nonnegative.
Note that {after} time $\rho$, both $u$ and $v$ know the rumour.
\end{definition}

\begin{observation}
\label{obs:comtime}
Let $uv \in E(G)$.
In the synchronous version,
$$T(uv) \samedist
1 + \min \{\geo(1/\deg(u)),\geo(1/\deg(v))\}
\:.
$$
\end{observation}

Using
Observations~\ref{obs:thinning} and~\ref{obs:restart},
we obtain a nicer formula for the asynchronous version.

\begin{proposition}
\label{pro:comtime}
Let $uv \in E(G)$.
In the asynchronous version,
\begin{equation}
\label{eqtuv}
T(uv) \samedist
\erv(1/\deg(u) + 1/\deg(v))\:.
\end{equation}
Moreover, the random variables $\{T_e\}_{e\in E(G)}$ are mutually independent.
\end{proposition}
\begin{proof}
By Observations~\ref{obs:thinning} and~\ref{obs:restart}, the $T(e)$'s are mutually independent, and moreover, $T(uv)$ is the minimum of two independent exponential random variables with rates $1/\deg(v)$ and  $1/\deg(u)$.
\end{proof}

We next prove that changing the starting vertex affects the spread time by at most a multiplicative factor of 2.

\begin{proposition}
\label{pro:start}
For any two vertices $u$ and $v$ of $G$ we have
$\sts(G,u) \stochle 2 \sts(G,v)$
and also
$\sta(G,u) \stochle 2 \sta(G,v)$.
\end{proposition}

\begin{proof}
We first consider the synchronous protocol.
Let $C(u,v)$ denote the first time that $v$ learns the rumour, assuming initially only $u$ knows it.
We claim that
\begin{equation}
C(u,v) \samedist C(v,u) \:,
\label{uvvusame}
\end{equation}
which would imply
$$
\sts(G,u)
\stochle 
C(u,v) + \sts(G,v)
\samedist
C(v,u) + \sts(G,v)
\stochle 2 \sts(G,v) \:.$$
In every round of an execution of the protocol, each vertex contacts a neighbour.
We call this an \emph{action}, and the 
\emph{signature} of this action, is a function $a : V \to V$ mapping each vertex to a neighbour.
Hence, $m$ rounds of the protocol can be 
encoded as $(u,a_1a_2\cdots a_m)$,
where $u$ is the vertex knowing the rumour initially, and
$a_1 a_2 \cdots a_m$ is a sequence of signatures.
Let $I(u,a_1 a_2 \cdots a_m)$ denote the set of informed vertices after $m$ rounds.
Note that in each round, the signature of the  {action taken} is a uniformly random one.
Hence $\p{C(u,v) \le k}$ equals the proportion of the signature-sequences $a_1 a_2 \cdots a_k$ of length $k$
that satisfy $v \in I(u,a_1 \cdots a_k)$.
If $v \in I(u,a_1 \cdots a_k)$, then looking at the $(u,v)$-path through which $v$ was informed, we see that $u \in I(v,a_k a_{k-1} \cdots a_2 a_1)$.
Therefore, $\p{C(u,v) \le k} = \p{C(v,u) \le k}$ for any $k$,
and this proves \eqref{uvvusame}.

We now consider the asynchronous protocol.
Let $D(u,v)$ denote the first time that $v$ learns the rumour, assuming initially only $u$ knows it.
Again, it suffices to prove
\begin{equation}
D(u,v) \samedist D(v,u) \:.
\label{uvvusamea}
\end{equation}
By Proposition~\ref{pro:comtime},
for any edge $uv$
we have $T(uv)\samedist \erv(1/\deg(u) + 1/\deg(v))$.
Moreover, the variables $\left\{T(e)\right\}_{e\in E}$
are mutually independent.
We define a collection of mutually independent random variables $\left\{R(e)\right\}_{e\in E}$,
such that for any edge $uv$, 
$$R(uv)\samedist \erv(1/\deg(u) + 1/\deg(v))\:.$$
Let $\mathcal P$ denote the set of all $(u,v)$-paths.
Then we have
$$D(u,v) 
=
\min \left\{ \sum_{e \in P} T(e) : P \in \mathcal P \right\}
\samedist \min \left\{ \sum_{e \in P} R(e) : P \in \mathcal P \right\} \:.$$
By symmetry, $D(v,u)$ has exactly the same distribution,
and \eqref{uvvusamea} follows.
\end{proof}

We next study some important graphs and bound their spread times, partly for showing tightness of some of the bounds obtained, and partly to serve as an introduction to the behaviour of the protocols. 
\subsection{The complete graph}
For the complete graph, $K_n$, by symmetry what matters at any time is not the actual set of informed vertices, but only the number of
vertices that have the rumour.
In the asynchronous case, by Proposition~\ref{pro:comtime}
and Observation~\ref{obs:thinning}, we can imagine a exponential clock for
each edge, having rate $2/(n-1)$ and independent of all other clocks. 
Let $T_1=0$ and denote by $T_k$
the first time that there are $k$ informed vertices. 
We can at this time simply restart all $k(n-k)$ clocks   at edges joining informed to uninformed vertices
(see Observation~\ref{obs:restart}). 
When the next alarm rings, a new vertex
receives the rumour.   
Thus, $T_{k+1}-T_k$ is
distributed as the mimimum of $k(n-k)$ independent exponential random variables each with rate $2/(n-1)$,
i.e.\ as $\erv(2k(n-k)/(n-1))$.
Hence by linearity of expectation,
\[
\wastapp(K_n)=
\e{\sta(K_n,v)}=\e{T_n }
 = \e{T_1} +
 \sum_{k=1}^{n-1} \e{T_{k+1}-T_k} =
 \sum_{k=1}^{n-1}\frac{n-1}{2k(n-k)}\:.
\]
We have
$$
\sum_{k=1}^{n-1}\frac{n-1}{2k(n-k)}
=\left(\frac{n-1}{2n}\right)
\sum_{k=1}^{n-1}\left\{\frac{1}{k}+\frac{1}{n-k}\right\}
\sim 
\left(\frac{n-1}{2n}\right) (2\log n) \sim \log n \:,
$$
so $\wastapp(K_n)\sim\log n$. 
In fact, Janson~\cite[Theorem~1.1(ii)]{asynchronous_complete} showed that a.a.s.\ $\sta(K_n,v) \sim \log n$.
Moreover, by slightly altering his proof we get $\gstapp(K_n)\sim (3/2)\log n$.

For the synchronous version, Karp et~al.~\cite[Theorem~2.1]{KSSV00} showed that $\sts(K_n,v)
\sim \log_3n$ a.a.s.
It follows that $\wastspp(K_n) \sim \log_3 n$. It is implicit in their proof that
$\gstspp(K_n) = O(\log n)$.

\subsection{The star}
\label{sec:star}
The star  $\sg$ with $n$ vertices has $n-1$ leaves and a central vertex that is adjacent to every other vertex.
It is clear that $\sts(\sg,v)=1$ if $v$ is the central vertex and $\sts(\sg,v)=2$ otherwise. 
So we have $\wastspp(\sg)=\gstspp(\sg)=2$.
Below we will show that  
$\wastapp(\sg) \sim \log n$ and $\gstapp(\sg)\sim 2\log n$.
This graph gives that the left-hand bounds in~\eq{expectedguaranteed},~\eq{expectedextremal},~\eq{guaranteedextremal},~\eq{expectedguaranteeds} and Corollary~\ref{cor:ratio}, and Theorem~\ref{thm:allring}, are tight, up to constant factors. 

We now show that $\wastapp(\sg) \sim \log n$.
The intuition is that in the asynchronous case, the spread time is close to the time the last vertex makes its first call.
By Proposition~\ref{pro:comtime},
all communication times are independent and distributed as $\erv(n/(n-1))$. 
Let $X_1,\dots,X_{n-1}$ be independent $\erv(n/(n-1))$ random variables. Then  
\[
\sta(\sg,v) \,{\buildrel d \over =}\, \begin{cases}
\max \{ X_1,\dots,X_{n-1} \} & \text{ if } v \text{ is the central vertex} \\
X_1+\max \{ X_2,\dots,X_{n-1} \} & \text{ if } v \text{ is a leaf.} 
\end{cases}
\]
It follows that $\wastapp(\sg) \sim \log n$ and a.a.s.\ $\sta(\sg,v)\sim \log n $ for any $v$. 

We finally show that $\gstapp(\sg)\sim 2\log n$.
Let $\lambda=n/(n-1)$, $Y=X_1$ and
define $Z=\max\{X_2,\dots,X_{n-1}\}$. Fix $\varepsilon\in (0,1)$. We have
\[
\p{Z\le \log n}= (1-e^{-\lambda \log n})^{n-2} \sim 1/e.
\]
Hence
\begin{align}\label{lower}
\p{Y+Z\ge (2-\varepsilon)\log n} \ge \p{Y\ge (1-\eps)\log n}\, \p{Z \ge \log n} = \Theta(n^{-1+\eps}).
\end{align}
Now let $a=(1+\eps/2)\log n$ and let $A$ be the event $\{Y\le a\}$.
We have $\p{A^c}=O(n^{-1-\eps/2})$. Also
\begin{align}\label{Y+Z}
\p{Y+Z\ge (2+\varepsilon)\log n} &\le \p{A^c}+\int_{0}^{a} \lambda e^{-\lambda y}\, \p{Z>(2+\eps)\log n-y}\mathrm d y		\notag	\\
&=O(n^{-(1+\eps/2)})+ \int_{0}^{a} \lambda e^{-\lambda y}\, \p{Z>(2+\eps)\log n-y} \mathrm d y.
\end{align}
Using the independence of the $X_i$,
\begin{align}\label{Z}
\p{Z>(2+\eps)\log n-y} &= 1- \left(1-e^{-\lambda[(2+\eps)\log n-y]}\right)^{n-2}	\notag	\\
&\le 1- \left(1-e^{-\lambda[(2+\eps)\log n-y]}\right)^{n}	\notag		 \\
&= 1- \left(1-\frac{e^{\lambda y}}{n^{\lambda(2+\eps)}}\right)^{n} \le \frac{e^{\lambda y}}{n^{1+\eps}}.
\end{align}
The last inequality can be justified by expanding the left hand side and using the fact that $\lambda>1$. Using \eqref{Z} in \eqref{Y+Z} we get
\[
\p{Y+Z\ge (2+\varepsilon)\log n} \le O(n^{-(1+\eps/2)})+ \frac{1}{n^{1+\eps}}\int_{0}^{a} \lambda e^{-\lambda y}\,e^{\lambda y}dy =O\left(n^{-(1+\eps/2)}\right).
\]
This equation together with equation~\eqref{lower} implies  $\gstapp(G_n^*)\sim 2\log n$.

\subsection{The path}
\label{s:path}
 
For the path graph $\pg$,
we have
$\wastapp(\pg)\sim n$, which shows that the right-hand bound in~\eq{expectedextremal} is tight, up to the constant factor.
Moreover,
$\gstapp(\pg)\sim n$.
For the synchronous protocol, we have $\wastspp(\pg)= (4/3)n-2$, which shows that the right-hand bound in~\eq{expectedextremals} is tight, up to the constant factor.
Finally, we have $\gstspp(\pg)\sim (4/3)n$.
Detailed calculations follow.

Label the vertices in the path as $(v_1,\dots,v_n)$. 
In this graph, the spread times in the synchronous and asynchronous variants are close to each other. We first consider the asynchronous variant.
Let $e$ be an edge.
By Proposition~\ref{pro:comtime}, if $e$ connects two internal vertices, then 
$T(e)\samedist \erv(1)$,
and otherwise, $T(e)\samedist\erv(3/2)$. 
Thus if the rumour starts from one of the endpoints, say $v_1$, we have
\begin{equation}
\label{pathgraph}
\sta(\pg,v_1)\,{\buildrel d \over =}\, \sum_{i=1}^{n-1} X_i,
\end{equation}
where $X_i$'s are independent exponential random variables, $X_1$ and $X_{n-1}$ with rates $3/2$ and the rest with rates $1$.
It follows that $\e{\sta(\pg,v_1)}= (n-3)+2(2/3)=n-5/3$. With similar computations, it is easy to see that this is the worst case, i.e.\ $\wastapp(\pg)= n-5/3$.

Next we show $\gstapp(\pg)\sim n$.
Fix $\eps>0$.
Note that
$\sum_{i=2}^{n-2} X_i$
is a sum of i.i.d.\ random variables,
hence by Cram\'{e}r's Theorem
(see, e.g.,~\cite[Theorem~5.11.4]{grimmett_book}),
the probability that it deviates by at least $\eps n$ from its expected value is $\exp(-\Omega(n))$.
Moreover,~\eq{pathgraph}
means that $\sta(\pg,v_1)$ is 
$\sum_{i=2}^{n-2} X_i$
plus two exponential random variables with constant rate,
and the same statement is true for it as well, so $\gstapp(\pg)\sim n$.
 
Now consider the synchronous case.
Let $e$ be an edge.
By Observation~\ref{obs:comtime}, if $e$ connects two internal vertices, then 
$T(e)\samedist 1+\geo(3/4)$,
and otherwise, $T(e)=1$. 
Thus if the rumour starts from one of the endpoints, say $v_1$, we have
\begin{equation}
\label{pathgraphs}
\sts(\pg,v_1)\samedist n-1 + \sum_{i=2}^{n-2} X_i,
\end{equation}
where $X_i$'s are independent
$\geo(3/4)$ random variables.
{It follows that $\e{\sts(\pg,v_1)}= (4/3)n-2$. 
With similar computations, it is easy to see that this is the worst case, i.e.\ $\wastspp(\pg)= (4/3)n-2$.}
An argument similar to the one for the asynchronous variant gives $\gstspp(\pg)\sim (4/3) n$.

\subsection{The double star}
\label{sec:twostars}
Consider the tree $DS_n$ consisting of two adjacent vertices of degree $n/2$ and $n-2$ leaves, see Figure~\ref{fig:stringdiamonds}(Left).
Below we will show that $\gstapp(DS_n)$ and $\gstspp(DS_n)$ are both  $\Theta(n \log n)$, while the average times $\wastapp(DS_n)$ and $\wastspp(DS_n)$ are $\Theta(n)$.
This example hence shows   tightness of the right-hand bounds in~\eq{expectedguaranteed},~\eq{guaranteedextremal},~\eq{expectedguaranteeds} and~\eq{guaranteedextremals} {up to constant factors}.
The main delay in spreading the rumour in this graph comes from the edge  joining the two centres. 
The idea is that it takes $\Theta(n)$ units of time on average for this edge to pass the rumour, but to be sure that this has happened with probability $1-1/n$, we need to wait $O(n\log n)$ units of time.
Detailed calculations follow.

Indeed we will show that
$\wastapp(DS_n)$ and $\wastspp(DS_n)$ are asymptotic to $n/4$,
and 
$\gstapp(DS_n)$ and $\gstspp(DS_n)$ are asymptotic to $n\log n / 4$.
First, consider the asynchronous case. Here, by Proposition~\ref{pro:comtime},
$T(e^*)= \erv(4/n)$.
So, the rumour passes from one centre to the other one in $n/4$ time units on average. 
On the other hand, the leaves learn the rumour in $\Theta(\log n)$ time on average, as in the star graph. 
Combining the two, we get $\wastapp(DS_n)\sim n/4$. 

For the guaranteed spread time, note that 
if $c<1/4$ then
$$\p{T(e^*)\ge cn\log n} 
=\exp(-n/4 \times cn \log n)\ge 1/n \:.$$ 
Thus $\gstapp(DS_n) \ge n\log n/4$. 
Straightforward calculations  give that 
if $c>1/4$ then for any vertex $v$,
$\p{\sta(DS_n,v) > cn\log n}<1/n$,
whence $\gstapp(DS_n)\sim (n\log n)/4$. 

In the synchronous case, for any $v$ we have
$T(e^*) + 1 \le \sts(DS_n,v) \le T(e^*) + 2$
and by Observation~\ref{obs:comtime},
$$T(e^*) \samedist
1 + \min \{\geo(2/n),\geo(2/n)\}
\samedist
1 + \geo(4/n - 4/n^2) \:,
$$
hence
$\wastspp(DS_n) = 3 + \e{\geo(4/n - 4/n^2)} \sim n/4$.
If $c<1/4$ then
$$\p{3+\geo(4/n - 4/n^2)\ge cn\log n} 
\ge
(1-4/n +4/n^2)^{cn\log n} \:.$$
Since $e^{-y}\ge 1-y\ge e^{-y-y^2}$ for every $y\in[0,1/4]$,
$$
(1-4/n +4/n^2)^{cn\log n}
=
\exp( (-4/n+O(1/n^2))cn\log n )
=
(1/n)^{4c}
e^{o(1)}
\ge 1/n \:.$$ 
While, if $c>1/4$, then
$$\p{3+\geo\left(\frac 4 n - \frac {4}{n^2}\right)\ge cn\log n} 
=
\left(1-\frac 4 n + \frac {4}{n^2}\right)^{cn\log n-3} \! \! \!
\le
\exp(-4c \log n + o(1))
< 1/n \:,$$ 
whence
$\gstspp(DS_n)\sim (n\log n)/4$.

\subsection{The string of diamonds}
\label{sec:chain}
Let $m$ and $k\ge 2$ be positive integers, and} let $\mathcal{S}_{m,k}$ be the `{string of diamonds}' graph {given in Figure~\ref{fig:stringdiamonds}(Right),} where there are $m$ diamonds, each  consisting of ${k}$ edge-disjoint paths of length $2$ with the same end vetices.
A vertex with degree greater than two is called a \emph{hub}. 
There are $m+1$ hubs and $km$ non-hubs, giving a total of $n=km+m+1$ vertices.
It turns out that in this graph the asynchronous push\&pull protocol is much quicker than its synchronous variant.

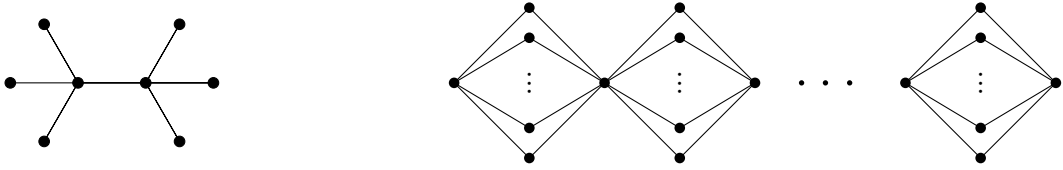
\begin{figure}
\centerline{\begin{tikzpicture}
\coordinate (v1) at (0,0) {};
\fill (v1) circle[radius=2pt] node {};
\coordinate (v2) at (1,1) {};
\fill (v2) circle[radius=2pt] node {};
\coordinate (v3) at (1,0.6) {};
\fill (v3) circle[radius=2pt] node {};
\coordinate (v4) at (1,-1) {};
\fill (v4) circle[radius=2pt] node {};
\coordinate (v6) at (2,0) {};
\fill (v6) circle[radius=2pt] node {};
\coordinate (v5) at (1,-.6) {};
\fill (v5) circle[radius=2pt] node {};
\coordinate (v8) at (3,1) {};
\fill (v8) circle[radius=2pt] node {};
\coordinate (v9) at (3,0.6) {};
\fill (v9) circle[radius=2pt] node {};
\coordinate (v10) at (3,-1) {};
\fill (v10) circle[radius=2pt] node {};
\coordinate (v11) at (4,0) {};
\fill (v11) circle[radius=2pt] node {};
\coordinate (v12) at (3,-.6) {};
\fill (v12) circle[radius=2pt] node {};
\coordinate (v18) at (6,0) {};
\fill (v18) circle[radius=2pt] node {};
\coordinate (v13) at (7,1) {};
\fill (v13) circle[radius=2pt] node {};
\coordinate (v14) at (7,0.6) {};
\fill (v14) circle[radius=2pt] node {};
\coordinate (v15) at (7,-1) {};
\fill (v15) circle[radius=2pt] node {};
\coordinate (v16) at (8,0) {};
\fill (v16) circle[radius=2pt] node {};
\coordinate (v17) at (7,-.6) {};
\fill (v17) circle[radius=2pt] node {};
\foreach \from/\to in {v18/v13, v18/v14,v18/v15,v18/v17, v17/v16, v13/v16,v14/v16,v15/v16}
  \draw (\from) -- (\to);
\path (v3) -- (v5) node [font=\large, midway, sloped] {...};
\path (v9) -- (v12) node [font=\large,midway, sloped] {...};
\path (v11) -- (v18) node [font=\huge, midway, sloped] {$\dots$};
\path (v14) -- (v17) node [font=\large, midway, sloped] {...};
\foreach \from/\to in {v1/v2, v1/v3, v1/v4, v2/v6, v3/v6,v4/v6,v1/v5,v6/v5, v6/v8,v6/v9,v6/v10,v6/v12, v11/v8,v11/v9,v11/v10,v11/v12}
  \draw (\from) -- (\to);
\begin{scope}[xshift=-5cm,scale=0.45]
\node at (0,0) [vtx]{};
\node at (-2,0) [vtx]{};
\node at (2,0) [vtx]{};
\node at (4,0) [vtx]{};
\node at (2+1,1.732) [vtx]{};
\node at (-1,1.732) [vtx]{};
\node at (2+1,-1.732) [vtx]{};
\node at (-1,-1.732) [vtx]{};
\draw(-2,0)--(0,0)
--(2,0)--(2+1,1.732)--(2,0)--(4,0)--(2,0)
--(2+1,-1.732)--(2,0)--(0,0)
--(-1,-1.732)--(0,0)
--(-1,1.732)--(0,0);
\end{scope}
\end{tikzpicture}}
\caption{Left: the double star graph $DS_8$,
which has a large guaranteed spread time in both variants.
Right: a string of diamonds, on which the asynchronous push\&pull protocol is much quicker than its synchronous variant.}
\label{fig:stringdiamonds}
\end{figure} 

Let us analyze the average spread times in the two protocols, starting with the asynchronous case.
Proposition~\ref{pro:comtime} gives that
for each edge $e$,
$$T(e)\samedist \erv(1/2 + 1/k) \stochle  \erv(1/2)$$
and that $\{T(e)\}_{e\in E}$ are independent. 
Between any two consecutive hubs there are $k$ disjoint paths of length 2,
so the communication time between them is stochastically dominated by $Z:=\min\{Z_1,\dots,Z_k\}$, where the $Z_i$  are independent random variables equal in distribution to the sum of two independent $\erv(1/2)$ random variables. 
\begin{lemma}
\label{lem:z}
We have
$\e{Z} = O(1/\sqrt k)$.
\end{lemma}
\begin{proof}
For any $t\geq 0$ we have
$$\p{Z>t}=\prod_{i}\p{Z_i>t}
=\p{Z_1>t}^k
\le
\left(1-\p{\erv(1/2)\le t/2}^2 \right)^k
= \left(2e^{-t/4}-e^{-t/2}\right)^k\:.$$
Thus, using the inequality
$2e^{-t/4}-e^{-t/2} \le e^{-t^2/64}$
valid for $t\in[0,4]$, we find
$$
\e{Z}=\! \int_{0}^{\infty}  \!  \p{Z> t}\mathrm d t \le
\int_{0}^{4} e^{-kt^2/64} \mathrm{d} t +
\int_{4}^{\infty} (2e^{-t/4})^k \mathrm{d} t
\le
8 \sqrt{\pi / k} + \frac{2^{k+2}}{ke^k} = O(1/\sqrt k) \:.\qedhere$$\end{proof}

By Lemma~\ref{lem:z}, the expected time for all the hubs to learn the rumour is $O(mk^{-1/2})$. 
Once all the hubs learn the rumour, a degree 2 vertex pulls the rumour in Exp$(1)$ time and the expected value of the maximum of at most $km$ independent  Exp$(1)$ variables is $O(\log km)$. 
So by linearity of expectation,
$\wastapp(G)=O(\log n+mk^{-1/2})$.

In the synchronous case, for any $G$ we have $\wastspp(G)\ge\diam(G)$.
For this graph, we get
$\wastspp(G)\ge 2m$.
Choosing $k=\Theta\left((n/\log n)^{2/3}\right)$ and $m = \Theta\left(n^{1/3}(\log n)^{2/3}\right)$ gives 
$$\wastapp(G)= O(\log n) \mathrm{\ and\ } \wastspp(G) =\Omega(n^{1/3}(\log n)^{2/3}) \:.$$
This graph  has $\wastspp(G) / \wastapp(G) = \Omega\left( (n/\log n)^{1/3}\right)$ and is the example promised by Corollary~\ref{cor:ratio}.

\section{Extremal spread times for push\&pull protocols}
\label{sec:thm:extremal}
In this section we prove
Theorems~\ref{thm:extremal}
and~\ref{thm:extremalsynchronous}.
\subsection{Proof of~\eq{expectedguaranteed} and its tightness}
For a given $t\ge0$, consider the protocol which is the same  as push\&pull  except that,  if the rumour has not spread to all vertices by   time $t$, then the new process  reinitializes. Coupling the new process with push\&pull,  we obtain for any $k\in\{0,1,2,\dots\}$ that
\begin{equation}\label{eq:coupling}
\p{\sta(G,v) > kt} \le \p{\sta(G,v) > t}^k.
\end{equation}
and
\begin{equation}\label{eq:couplings}
\p{\sts(G,v) > kt} \le \p{\sts(G,v) > t}^k.
\end{equation}
Combining \eq{eq:coupling} with  
\[
\p{\sta(G,v) > e \e{\sta(G,v)}} < 1/e,
\]
 which comes directly from Markov's inequality,  we obtain
\[
\p{\sta(G,v) > e  \log n\,\e{\sta(G,v)}}  < 1/n.
\]
Since $\e{\sta(G,v)}\le \wastapp(G)$ for all $v$, this gives the right-hand inequality   in \eq{expectedguaranteed} directly from the definition of $\gstapp$. 
{This inequality is tight up to the constant factor, as the double star has $\wastapp(DS_n) = \Theta(n)$ and
$\gstapp(DS_n) = \Theta(n \log n)$
(see Section~\ref{sec:twostars}).}

To prove the left-hand inequality, let $\tau= \gstapp(G)$ and let $v$ be a vertex such that $\e{\sta(G,v)}=\wastapp(G)$. Then
\begin{align*}
\wastapp(G)&= \int_{0}^{\infty}\p{\sta(G,v)>t}dt= \sum_{i\in\Nz}\int_{i\tau}^{(i+1)\tau}\p{\sta(G,v)>t}dt  
 \le \sum_{i\in\Nz} \frac{\tau }{n^i}
\end{align*}
by~\eqref{eq:coupling} with $t=\tau$. Hence $\wastapp(G)\le \tau/(1-1/n)$. 
{
This inequality is tight up to a constant factor,
as the star has $\wastapp(G^*_n) = \Theta(\gstapp(G^*_n)) = \Theta (\log n)$ (see Section~\ref{sec:star}).}

\subsection{Proof of the right-hand bound in~\eq{expectedextremal} and its tightness}
\label{sec:on}
We will actually prove this using pull operations only.
Indeed we will show    $\wastapull(G) < 4n$, where the superscript $\mathsf {pull}$ means the `pull only' protocol.
{Since the path has $\wastapull(P_n) \ge \wastapp(P_n) = \Theta(n)$
(see Section~\ref{s:path}), this bound would be tight up to the constant factor.}

The proof is by induction: we prove that when there are precisely $m$ uninformed vertices, just $b$ of which have informed neighbours (we call these {$b$} vertices the \emph{boundary vertices}), the expected  remaining   time for the rumour to reach all vertices is at most $ 4 m - 2b$.  The inductive step is proved as follows. Let $I$ denote the set of informed vertices, $B$   the set of boundary vertices,  and $R$   the set of the remaining vertices. 
Let $|B|=b$ and $|B|+|R|=m$.  Let $d(v)$ denote the degree of $v$ in $G$ and, for a set $S$ of vertices, let $d_S(v)$ count the number of neighbours of $v$ in   $S$. 
We consider two cases.

Firstly, suppose that there exists a boundary vertex $v$ with $d_R(v) \ge d_B(v)$.
We can for the next step ignore all calls from vertices other than   $v$,  so the process is forced to wait  until $v$ is informed before any other vertices. This clearly gives an upper bound on the spread time.
The expected time taken for $v$ to pull the rumour from vertices in $I$ is  
\[
{\frac{d(v)}{ d_I(v)}= \frac{d_I(v)+d_R(v)+d_B(v)}{d_I(v)}\le 
1+\frac{2d_R(v)}{d_I(v)}
\le 1+2d_R(v)}.
\] 
Once $v$ is informed, the number of uninformed vertices decreases by 1,
and the number of boundary vertices increases by $d_R(v)-1$. The inductive hypothesis  concludes this case since
 \[
1+2d_R(v)+4(m-1)-2(b+d_R(v)-1)< 4m-2b.
\] 
Otherwise, if there is no such $v$, then any boundary vertex $v$ has a `pulling rate' of 
\[
\frac{d_I(v)}{d_I(v)+d_R(v)+d_B(v)}   \ge \frac{1}{1+d_R(v)+d_B(v)}\ge \frac{1}{2d_B(v)}\ge \frac{1}{2b}.
\] 
Since there are $b$ boundary vertices, 
together they have a pulling rate of at least $1/2$, so
the expected time until a boundary vertex is informed is at most 2.
Once this happens,
$m$ decreases by 1 and $b$ either does not decrease or decreases by $1$, and the inductive hypothesis concludes the proof.

\subsection{Proof of the left-hand bound in~\eq{expectedextremal} and its tightness}
\label{sec:proofleftexpectedextremal}
In this section we show for any vertex $v_0$ {of a graph $G$} we have
$\e{\sta(G,v_0)}>(\log n)/5$.
This is tight (up to the constant) as the star  has $\wastapp(G^*_n) = O(\log n)$ (see Section~\ref{sec:star}).
We give an argument for an equivalent protocol, defined below.

\begin{definition}[Two-clock-per-edge protocol]
On every edge place two exponential clocks, one near each end vertex. All clocks are independent. 
On an edge joining vertices $u$ and $v$, the clocks both have rate  $\deg(u)^{-1}+\deg(v)^{-1}$.
Note that this is the rate of calls along that edge, combined, from $u$ and $v$ (see Proposition~\ref{pro:comtime}).  
At any time that the clock near $u$ on an edge $uv$ rings,  and  $v$ knows the  rumour but $u$ does not, the rumour is passed to $u$. 
\end{definition}

\begin{proposition}
\label{pro:equivalent}
The two-clock-per-edge protocol is equivalent to the asynchronous push\&pull protocol.
\end{proposition}

\begin{proof}
Consider an arbitrary moment during the execution of the two-clock-per-edge protocol.
Let $I$ denote the set of informed vertices.
For any edge $uv$ with $u\in I$ and $v\notin I$,
the rate of calls along $uv$ is $\deg(u)^{-1}+\deg(v)^{-1}$. 
Moreover, the edges act independently.
So, the behaviour of the protocol at this moment is exactly the same as that of the asynchronous push\&pull protocol.
Hence, the two protocols are equivalent.  
\end{proof}

In view of Proposition~\ref{pro:equivalent},
we may work with the two-clock-per-edge protocol instead.
Let $X_v$ be the time taken for the first clock located near $v$ to ring. 
Then $X_v$ is  distributed as  $\erv(f(v))$ where
$f(v)  = 1 + \sum \deg(u)^{-1}$, the sum being over all neighbours $u$ of $v$.
Hence,
$\sum f(v) = 2n$.

On the other hand, for a vertex $v\neq v_0$ to learn the rumour, at least one of clocks located near $v$ must ring. 
Thus 
$$\max \{ X_v : v \in V(G){\setminus \{v_0\}}\} \stochle \sta(G,v_0) \:.$$
Let $X=\max \{ X_v : v \in V(G){\setminus \{v_0\}}\}$.
Hence to prove
$\e{\sta(G,v_0)}>(\log n)/5$
it suffices to show 
$\e{X} > (\log n)/5$.

Let $\tau = \log(n-1)/3$
and $A = V(G)\setminus\{v_0\}$,
Then we have 
\begin{align*}
\p{X < \tau} &= 
\prod_{v\in A} (X_v < \tau)
= \prod_{v\in A} (1-e^{-\tau f(v)})
\le \exp\left(- \sum_{v\in A} e^{-\tau f(v)}\right)
\\&\le  \exp\left(-(n-1)e^{-\tau\sum_vf(v)/(n-1)}\right) 
\le  
\exp\left(-(n-1)e^{-3\tau}\right) =e^{-1} \:.
\end{align*}
Here the first inequality follows from $1-x\le e^{-x}$, the second from the arithmetic-geometric mean inequality, and the last one from 
$2n=\sum_v f(v) \le 3(n-1)$
which holds for $n\ge3$.
Consequently,
$$
\e{X}\ge \p{X\ge \tau} \tau
\ge (1-e^{-1}) \log(n-1)/3 > \frac15 \log n $$
for all large enough $n$.

\subsection{Proof of~\eq{guaranteedextremal} and  its tightness}
The bounds in~\eq{guaranteedextremal} follow immediately from~\eq{expectedguaranteed} and~\eq{expectedextremal}. 
The left-hand bound is tight as the star has
$\gstapp(G^*_n) = \Theta(\log n)$
(see Section~\ref{sec:star}),
and the right-hand bound is tight as the double star has
$\gstapp(DS_n) = \Theta(n \log n)$
(see Section~\ref{sec:twostars}).

\subsection{Proof of~\eq{expectedextremals} and its tightness}
\label{apex}
In this section we will prove $\wastspp (G) <4.6n $, which is tight up to the constant factor, as the path has diameter $n-1$ and hence $\wastspp(P_n) \ge n-1$.

The proof is similar to the one for the right-hand bound in~\eq{expectedextremal} given in Section~\ref{sec:on}. 
Let $\alpha= { \sqrt{e}/(\sqrt{e}-1)}$.
We consider the `pull only' protocol,
and will prove inductively that when there are $m$ uninformed vertices and $b$ boundary vertices, the expected remaining time for the rumour to reach all vertices is at most $(2+\alpha)m-2b$, and it follows that $\wastspp (G) < 4.6 n$.
The inductive step is proved as follows. Let $I$ denote the set of informed vertices, $B$   the set of boundary vertices,  and $R$   the set of the remaining vertices. 
Let $|B|=b$ and $|B|+|R|=m$.  Let $d(v)$ denote the degree of $v$ in $G$ and, for a set $S$ of vertices, let $d_S(v)$   denote the number of neighbours of $v$ in   $S$. 
Consider two cases.

Firstly, suppose that there is a vertex $v\in B$ such that $d_R(v)\ge d_B(v)$ (see Figure~\ref{fig:case1}(Left)).
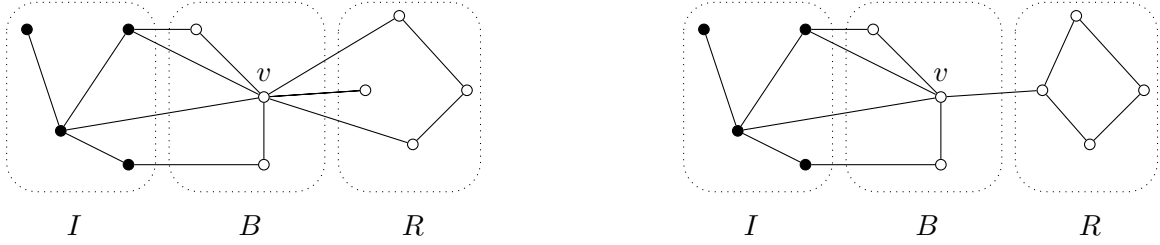
\begin{figure}
\centering
\begin{tikzpicture}[scale=0.9]
\node (f) at (0,0)[evtx,label=above:$v$]{};
\node (e) at (-1,1)[evtx]{};
\node (g) at (0,-1)[evtx]{};
\node (b) at (-2,1)[vtx]{};
\node (d) at (-2,-1)[vtx]{};
\node (c) at (-3,-0.5)[vtx]{};
\node (a) at (-3.5,1)[vtx]{};
\node (h) at (2,1.2)[evtx]{};
\node (i) at (1.5,0.1)[evtx]{};
\node (j) at (2.2,-0.7)[evtx]{};
\node (k) at (3,0.1)[evtx]{};
\draw(a)--(c)--(b)--(e)--(f)--(b);
\draw(c)--(f)--(g)--(d);
\draw(k)--(h)--(f)--(i)--(f)--(j)--(k);
\draw(c)--(d);

\draw [dotted, rounded corners = 4mm] (-1.4,1.4) rectangle (0.9,-1.4);
\node at (-0.2,-1.9){$B$};
\draw [dotted, rounded corners = 4mm] (-3.8,1.4) rectangle (-1.6,-1.4);
\node at (-2.8,-1.9){$I$};
\draw [dotted, rounded corners = 4mm] (1.1,1.4) rectangle (3.2,-1.4);
\node at (2.2,-1.9){$R$};

\begin{scope}[xshift=10cm]
\node (f) at (0,0)[evtx,label=above:$v$]{};
\node (e) at (-1,1)[evtx]{};
\node (g) at (0,-1)[evtx]{};
\node (b) at (-2,1)[vtx]{};
\node (d) at (-2,-1)[vtx]{};
\node (c) at (-3,-0.5)[vtx]{};
\node (a) at (-3.5,1)[vtx]{};
\node (h) at (2,1.2)[evtx]{};
\node (i) at (1.5,0.1)[evtx]{};
\node (j) at (2.2,-0.7)[evtx]{};
\node (k) at (3,0.1)[evtx]{};
\draw(a)--(c)--(b)--(e)--(f)--(b);
\draw(c)--(f)--(g)--(d);
\draw(k)--(h)--(i)--(j)--(k);
\draw(c)--(d);
\draw(f)--(i);

\draw [dotted, rounded corners = 4mm] (-1.4,1.4) rectangle (0.9,-1.4);
\node at (-0.2,-1.9){$B$};
\draw [dotted, rounded corners = 4mm] (-3.8,1.4) rectangle (-1.6,-1.4);
\node at (-2.8,-1.9){$I$};
\draw [dotted, rounded corners = 4mm] (1.1,1.4) rectangle (3.2,-1.4);
\node at (2.2,-1.9){$R$};
\end{scope}

\end{tikzpicture}
\caption{Left: first case in the proof of~\eq{expectedextremals}: there exists a boundary vertex $v$ with $3=d_R(v) \ge d_B(v)=2$.
Right: second case in the proof of~\eq{expectedextremals}: 
for all boundary vertices $v$ we have  $d_R(v) < d_B(v)$ 
(informed vertices are black, uninformed vertices are white).}
\label{fig:case1}
\end{figure}
In this case, for the next step, we ignore all calls from vertices other than $v$ and wait until $v$ is informed before any other uninformed vertex. This gives an upper bound on the spread time.
The expected time taken for $v$ to pull the rumour from vertices in $I$ is
\[ 
1+ \e{\geo\left(\frac {d_I(v)}{d(v)}\right)}=
\frac{d(v)}{d_I(v)}=\frac{d_I(v)+d_R(v)+d_B(v)}{d_I(v)} \le 2d_R(v)+1. 	\]
Once $v$ is informed, the number of uninformed vertices decreases by 1 and the number of boundary vertices increases by $d_R(v)-1$. By the inductive hypothesis the expected time for the spread of the rumour is at most
\[
2d_R(v)+1+(\alpha+2)(m-1)-2(b+d_R(v)-1)< (\alpha+2)m-2b.
\]

Next consider the case that $d_R(v)<d_B(v)$ for all $v\in B$ (see Figure~\ref{fig:case1}(Right)). 
For each boundary vertex $v$ we have
\[
\frac{d_I(v)}{d(v)}=
\frac{d_I(v)}{d_I(v)+d_R(v)+d_B(v)}   \ge \frac{1}{1+d_R(v)+d_B(v)}\ge \frac{1}{2d_B(v)}\ge \frac{1}{2b}.
\] 
Let $X$ denote the time taken until the next vertex is informed.
Then we have
$$X = 1 + \min \{X_1,X_2,\dots,X_b\}\:,$$
where the $X_i$'s are geometric random variables with parameters at least $1/2b$, and correspond to the waiting times of the boundary vertices,
and they are independent since we are considering pull operations only.
Thus we have
\begin{align*}
\e{X-1}
= \sum_{t\in\N}\p{X-1\ge t}
= \sum_{t\in\N}\prod_{i\in[b]} \p{X_i\ge t}
\le
\sum_{t\in\N}
\left(1-\frac{1}{2b}\right)^{tb}
\le
\sum_{t\in\N}
e^{-t/2}=\alpha-1\:,
\end{align*}
so a boundary vertex learns the rumour after at most $\alpha$ units of time on average, at which time the number of boundary vertices either does not decrease or decreases by 1. By inductive hypothesis again, the average spread time is at most
\[
\alpha+(\alpha+2)(m-1)-2(b-1)= (\alpha+2)m-2b,
\]
which completes the proof.

\subsection{Proof of~\eq{expectedguaranteeds} and~\eq{guaranteedextremals} and their tightness}
\label{sec:thm:extremalsynchronous}
The proof of~\eq{expectedguaranteeds} and its tightness are exactly the same as that for~\eq{expectedguaranteed}.
The bound $\gstspp (G) <4.6e n \log n$ is  a direct consequence of bounds~\eq{expectedextremals} and~\eq{expectedguaranteeds}. 
This bound is tight (up to the constant factor) as the double star has guaranteed spread time $\Theta(n\log n)$ (see Section~\ref{sec:twostars}).

\section{Comparison of the two protocols}
\label{sec:cor:ratio}
We first prove  Corollary~\ref{cor:ratio}
assuming Theorems~\ref{thm:allring}
and~\ref{thm:last},
and in the following subsections we prove
these theorems.
The left-hand bound in  Corollary~\ref{cor:ratio} follows from Theorem~\ref{thm:allring};
it is tight, up to the constant factor, as the star   has $\gstapp(G^*_n)=\Theta(\log n)$ and $\gstspp(G^*_n)=2$ (see Section~\ref{sec:star}).
The right-hand bound in  Corollary~\ref{cor:ratio} follows from
Theorem~\ref{thm:last} by choosing $\alpha=1/3$.
A graph $G$ was given in Section~\ref{sec:chain}
having
${\wastspp(G)}/{\wastapp(G)} = \Omega \left((n/\log n)^{1/3} \right )$.
Using~\eq{expectedguaranteed} and~\eq{expectedguaranteeds}, we get
${\gstspp(G)}/{\gstapp(G)} = \Omega \left(n^{1/3} (\log n)^{-4/3} \right)$ for this $G$.

\subsection{The lower bound}
In this section we prove
Theorem~\ref{thm:allring}.
Let $G$ be an $n$-vertex graph and let $s$ denote the vertex starting the rumour.
We give a coupling between the two versions.
Consider a `collection of calling lists for vertices': for every vertex $u$, we have an infinite list of vertices, each entry of which is a uniformly random neighbour of $u$, chosen independently from other entries, see Figure~\ref{fig:callinglist} for an example.
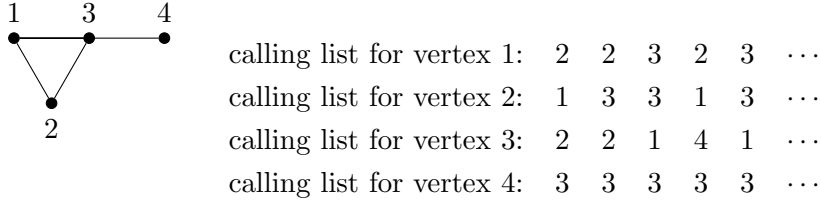
\begin{figure}
\centering
\begin{tabular}{l r}
\begin{tikzpicture}
\def \side {1cm}
\def \tside {2cm}
\node[vtx,label=above:1] at (0,0) (v2) {};
\node[vtx,label=below:2] at (-60:\side) (v1) {};
\node[vtx,label=above:3] at (0:\side) (v3) {};
\node[vtx,label=above:4] at (0:\tside) (v4) {};
\draw(v2)--(v3)--(v1)--(v2)--(v4);
\end{tikzpicture}
&
\begin{tabular}{l  c c c c c r}
calling list for vertex 1:& 2 & 2 & 3 & 2 & 3 & $\cdots$ \\
calling list for vertex 2:& 1 & 3 & 3 & 1 & 3 & $\cdots$ \\
calling list for vertex 3:& 2 & 2 & 1 & 4 & 1 & $\cdots$\\
calling list for vertex 4:& 3 & 3 & 3 & 3 & 3 & $\cdots$\\
\\
\end{tabular}
\end{tabular}
\caption{a particular outcome of the collection of calling lists for vertices}
\label{fig:callinglist}
\end{figure}

The coupling is built by using the same collection of calling lists for the two versions of the push\&pull protocol.
Note that $\sts(G,s)$ is determined by this collection,
but to determine $\sta(G,s)$ we also need to know the Poisson processes associated with the vertices.

We first prove that
$\wastapp (G) \leq  {H}_n \times \wastspp(G) $ (recall that ${H}_n$ denotes the $n$th harmonic number).
Consider  the asynchronous protocol and let $X_1$ be the first time such that all clocks have rung during the time interval $[0,X_1]$.
Let $X_2$ be the first time such that 
all clocks have rung during the time interval $(X_1,X_1+X_2]$,
and define $X_3,X_4,\dots$ similarly.
Partition $[0,\infty)$
into subintervals $[0,X_1]$, $(X_1,X_1+X_2]$,
$(X_1+X_2,X_1+X_2+X_3]$ etc.
Consider a `decelerated' variant $\mathcal{D}$ of the asynchronous push\&pull protocol in which each vertex makes a call the first time its clock rings in each subinterval, but ignores later clock rings in that subinterval (if any).
The spread time in $\mathcal{D}$ is stochastically larger than that in the asynchronous push\&pull protocol, so without loss of generality we may and will work with  $\mathcal{D}$.
Coupling $\mathcal{D}$ and the synchronous protocol using the same calling lists and using induction gives
$$\operatorname{ST}_{\mathcal{D}}(G,s)
\stochle X_1 + X_2 + \dots + X_{\sts(G,s)} \:.$$
Since the $X_i$ are i.i.d. and $\e{X_1} = {H}_n$,
Wald's equation (see, e.g.,~\cite[lemma~10.2.9]{grimmett_book}) gives
$$\e{\sta(G,s)}
\leq 
\e{\operatorname{ST}_{\mathcal{D}}(G,s)}
\leq {H}_n \times \e{\sts(G,s)} 
\leq {H}_n  \times \wastspp(G) \:,$$
as required.

Next we prove that $\gstapp (G) \leq 8  \gstspp(G) \log n $.
Let $B$ denote the event
`$\sts(G,s) \le 2 \gstspp(G)$', which depends on the calling lists only.
Inequality~\eq{eq:couplings} gives
$\p{B^c} \le 1/n^2$.
Partition the time interval $[0,2\gstspp(G) \times 4 \log n)$ into subintervals
$[0,4 \log n)$, $[4\log n,  8 \log n)$, etc.
Consider another `decelerated' variant $\mathcal{D}'$ of the asynchronous push\&pull protocol in which each vertex makes a call the first time its clock rings in each subinterval (if it does), but ignores later clock rings in that subinterval (if any).
The spread time in $\mathcal{D}'$ is stochastically larger than that in the asynchronous push\&pull protocol, so without loss of generality we may and will work with $\mathcal{D}'$.
Let $A$ denote the event
`during each of these $2\gstspp(G)$ subintervals, all clocks ring at least once.'
If $A$ happens, then an inductive argument gives that for any $1\le k \le 2\gstspp(G)$, the set of informed vertices in the $\mathcal{D}'$ at time $4k\log n$ {contains} the set of informed vertices after $k$ rounds of the synchronous version.
Hence, if both $A$ and $B$ happen,
then we would have
$$\sta(G,s) \le 
\operatorname{ST}_{\mathcal{D}'}(G,s) \leq 
(4\log n) \sts(G,s) \le
( 8 \log n) \gstspp(G) \:.$$
Hence to complete the proof, we need only show that
$\p {A^c} \le 1/n - 1/n^{2}$.
 
Let $I$ denote a given subinterval.
In the asynchronous version, the clock of any given vertex rings with probability at least $1-n^{-4}$ during $I$.
By the union bound, all clocks ring at least once during $I$, with probability at least $1-n^{-3}$.
The number of subintervals in the definition of $A$ is
$2\gstspp(G)$, which is $O(n\log n)$ by~\eq{guaranteedextremals}.
By the union bound again,
$\p{A^c} =O\left( \log n / n^{2} \right)$, as required.

Theorem~\ref{thm:allring} is tight, up to the constant factors, as the star has $\gstapp(G^*_n)=\Theta(\log n)$ and $\gstspp(G^*_n)=2$ (see Section~\ref{sec:star}).

\subsection{The upper bound}
In this section we prove Theorem~\ref{thm:last}.
Fix $\alpha \in [0,1)$. We want to prove
\begin{equation}\label{brv2}
{\gstspp(G)} \le  n^{1-\alpha} + 64 {\gstapp(G)} n ^{(1+\alpha)/2}  \:.
\end{equation}
We first sketch the proof.
The main ingredients in the proof are a coupling between the two protocols, and sharp concentration bounds. Consider the asynchronous version. List the vertices in the order their clocks ring. The list ends once all the vertices {are informed}. Now consider the natural coupling between the two protocols, the synchronous actions follow the same {ordering} as in the list. 
We partition the list into blocks according to a certain rule {in such a way that the} blocks have the following property:
the synchronous {protocol} in each round will inform a superset of the set of vertices informed by the asynchronous {variant} in any single block.   
{For example,} if we require that in each block each vertex communicates with the others at most once, then we would have this property. However, in order to get our bound, we need to use a more {delicate} rule for building the blocks. To conclude, we find an upper bound for the number of blocks, which coincides with the right-hand side of \eqref{brv2}.

We now give the details.
{Let us fix an arbitrary} starting vertex.
Let $B_1,B_2,\dots$ be an i.i.d.\ sequence of vertices, where $B_i$ is a {uniformly random} vertex of $G$.
For each $i$, let $W_i$ be a uniformly random neighbour of $B_i$, chosen independently of all other choices.
Hence, $W_1,W_2,\dots$ is also an i.i.d.\ sequence of vertices (not necessarily having uniform distribution).
We define a coupling between the two protocols by using the two sequences $(B_i)_{i\in\N}$ and $(W_i)_{i\in\N}$.

To define the coupled asynchronous scenario, we also need to know the ringing times of the clocks.
Let $Z_1,Z_2,\dots$ be a sequence of i.i.d.\ exponentials with rate $n$ (and mean $1/n$), and let this sequence be independent of $(B_i)_{i\in\N}$ and $(W_i)_{i\in\N}$.
Then the coupled asynchronous scenario proceeds as follows:
at time $Z_1$ the clock of vertex $B_1$ rings and it contacts $W_1$,
then at time $Z_1+Z_2$ the clock of $B_2$ rings and it contacts $W_2$, and so on.

We now define a third rumour spreading scenario,
which corresponds to the so-called \emph{sequential protocol}~\cite{mixing}.\label{pageseq}
This protocol works as the asynchronous one except we put $Z_i=1$ for all $i$.
Hence, the scenario only depends on the sequences $(B_i)_{i\in\N}$ and $(W_i)_{i\in\N}$.
Let $N$ denote the first time that this protocol has informed all the vertices.
Note that $N\ge {n-1}$ and $N$ takes integral values.
Observe that, in the asynchronous scenario, all vertices are informed right after $N$ clocks have rung,
and the spread time is $\sum_{i=1}^{N}Z_i$.
The following lemma relates $\gstapp(G)$ and $N$.

\begin{lemma}
\label{lem:eventA}
Define the event 
$$\mathcal A := \{ N \le 4 n \gstapp(G) \} \:.$$
Then we have $\p{\mathcal A} \ge 1 - O(1/n^2)$.
\end{lemma}

\begin{proof}
As the spread time of the asynchronous scenario is $\sum_{i=1}^{N}Z_i$,
by definition of $\gstapp$ we have
$\p{\sum_{i=1}^{N}Z_i > \gstapp(G)} {\le} 1/n$.
By~\eqref{eq:coupling} we have
$$
\p{\sum_{i=1}^{N}Z_i > 2 \gstapp(G)} {\le} 1/n^2 \:.$$
So we need only show that
\begin{equation*}
\p{N > 2 \sum_{i=1}^{N} nZ_i } = O(1/n^2) \:.
\end{equation*}
Since $Z_i$'s are i.i.d.\ exponentials with rate $n$,
the random variables $nZ_{i}$ are i.i.d.\ exponentials with rate 1, so for any fixed $t$, 
Cram\'{e}r's Theorem gives
$$\p{\sum_{i=1}^{t} nZ_i <t/2} = e^{-ct}$$
for some positive constant $c$.
Since $N\ge n-1$, we have
$$\p{N > 2 \sum_{i=1}^{N} nZ_i } 
\le \sum_{t=n-1}^{\infty} 
\p{N > 2 \sum_{i=1}^{N} nZ_i \cond N = t} 
= \sum_{t=n-1}^{\infty} e^{-ct} = \exp(-\Omega(n))\:,
$$
as required.
\end{proof}

To define the coupled synchronous scenario, we 
need some definitions.
For each vertex $v$, let $\pi(v)$ denote the probability that $W_j = v$. 
Recall that this probability does not depend on $j$.  
Call a vertex $v$ \emph{special} if $\pi(v) > n^{\alpha-1}$. 
Note that 
since $\sum \pi(v)=1$,
there are less than $n^{1-\alpha}$ special vertices. 

We partition the list $B_1,W_1,B_2,W_2,\dots$ into infinitely many finite \emph{blocks} as follows.
The first block is of the form
$$B_1,W_1,B_2,W_2,\dots,B_j,W_j\:,$$
with $j$ as large as possible, subject to the following conditions:
\begin{enumerate}
\item
We have $B_i \notin \{B_1,W_1,\dots,B_{i-1},W_{i-1}\}$ 
for all $1< i \le j$.
\item
If $W_i \in \{B_1,W_1,\dots,B_{i-1},W_{i-1}\}$ for some $1< i \le j$,
then $W_i$ is special.
\end{enumerate}
Note that we choose the block to be as long as possible, hence we stop at $W_j$ only if $B_{j+1}$ already appears in $B_1,\dots,W_j$,
or $W_{j+1}$ is non-special and it appears in $B_1,\dots,W_j$, or both.
If we have stopped at $W_j$, then a new block is started from $B_{j+1}$, and this process is iterated forever to define all the blocks.  
Note that each block has an even number of elements.

Let $S_1,S_2,\dots$ denote the sizes of the blocks,
and let $N_b$ be the smallest number such that
$$S_1 + S_2 + \dots + S_{N_b} \ge 2N \:.$$
The following lemma relates the spread time of the synchronous protocol and $N_b$.

\begin{lemma}
\label{lem:couple_synch}
The spread time of the synchronous push\&pull protocol is stochastically {smaller than} $N_b + n^{1-\alpha}$.
\end{lemma}
\begin{proof}
In this proof we only consider the finite list
$B_1,W_1,\dots,B_{N},W_{N}$,
which is partitioned into blocks as discussed before.
We further split the blocks into smaller ones according to the following rule.
Let $v$ be a special vertex and assume that in the sequential scenario, it is informed exactly at time $i$.
So, either $B_i=v$ or $W_i=v$.
If the elements $B_i,W_i,B_{i+1},W_{i+1}$ are contained in the same block, then we split this block at this point, putting everything up to $B_i,W_i$ in one block and $B_{i+1},W_{i+1}$ and everything after in the other one.
Since the number of splits equals the number of special vertices,
and there are less than 
$n^{1-\alpha}$ special vertices,
the new total number of blocks is {less than}
$N_b + n^{1-\alpha}$.
We work with these refined blocks for the rest of the proof.

{We couple with a modified version of the synchronous push\&pull protocol, which we call the \emph{lazy} scenario.}
We define the coupled lazy scenario inductively using the blocks.
Assume that the $k$th block is
$$B_i,W_i,B_{i+1},W_{i+1},\dots,B_j,W_j \:.$$
Then in the $k$th round of the {lazy} scenario,
vertex $B_i$ contacts $W_i$, 
vertex $B_{i+1}$ contacts $W_{i+1}$ and so on, up until vertex $B_j$ contacts $W_j$
(all these communications happen at the same time).
Moreover, a vertex that does not appear in this block,
does not perform any action in the $k$th round.
It is clear that stochastic upper bounds for the spread time of this lazy {scenario} carries over to the {synchronous push\&pull scenario}.

To complete the proof we will show that
the set of vertices informed by the lazy scenario after $k$ rounds equals the set of vertices informed by the sequential scenario right after time $(S_1+\dots+S_k)/2$.
{(The factor of $2$ appears here because a block with $r$ communications has length $2r$.)} 
The proof proceeds by induction.
Assume that the $k$th block is
$$B_i,W_i,B_{i+1},W_{i+1},\dots,B_j,W_j \:.$$
If no repetition happens in this block at all,
then it is clear that {the} lazy {scenario} in one round informs every vertex which {the} sequential {one} informs during times $i,i+1,\dots,j$.
Notice the possible problem if a repetition happens:
if during this block, $x$ contacts $y$ and tells her the rumour for the first time,
and $z$ also contacts $y$ and asks her the rumour,
then in the sequential scenario both $y$ and $z$ will learn the rumour by time $j$,
whereas in the lazy scenario this is not the case because these operations happen at exactly the same time.
However, if $v$ is a repeated vertex in this block,
then $v$ is a special vertex,
and moreover by the secondary splitting of the blocks,
we know that it cannot be the case that 
$v$ is informed in this block for the first time
and appears again later in the block.
Hence, no `informing path' of length greater than one can appear in this block, and the proof is complete.
\end{proof}

Let $k = 64 \gstapp(G) n^{(1+\alpha)/2}$.
The following lemma bounds $N_b$.

\begin{lemma}
\label{lem:eventB}
Define the event 
$$\mathcal{B} := \{ S_1 + \dots + S_k \ge 8 \gstapp(G) n\} \:.$$
Then we have $\p{\mathcal{B}} \ge 1 - O(1/n^2)$.
\end{lemma}

Before proving this lemma, let us see why it concludes the proof of Theorem~\ref{thm:last}.
By Lemmas~\ref{lem:eventA} and~\ref{lem:eventB} and the union bound, with probability at least $1-1/n$
both events $\mathcal A$ and $\mathcal B$ happen.
Assume this is the case.
Then we have
$$S_1 + \dots + S_k \ge 8 \gstapp(G) n \ge 2N \:,$$
which means $N_b \le k$ by the definition of $N_b$.
Together with Lemma~\ref{lem:couple_synch},
this implies that with probability at least $1-1/n$, the spread time of the {synchronous push\&pull protocol} is at most $k + n^{1-\alpha}$, which gives~\eqref{brv2}.

\begin{proof}[Proof of Lemma~\ref{lem:eventB}.]
Let $\ell = n^{(1-\alpha)/2}/4$.
We first show that
\begin{equation}
\label{largeblocks}
\p{S_1 > 2\ell} \ge 1/2 \:.
\end{equation}
Let $j$ be arbitrary.
We compute the conditional probability of $\{S_1\ge 2j+2\}$ given that $\{S_1 \ge 2j\}$. 
On the event $\{ S_1 \ge 2 j\}$, the conditional  probability that $B_{j+1}$ is a repetition of a vertex already in the block $B_1,W_1,B_2,W_2,\dots,B_j,W_j$ is $2j/n$. 
The probability that $W_{j+1}$ is a repetition of a non-special vertex in the block is bounded above by 
$2j n^{\alpha -1}$, since there are at most $2j$ distinct vertices in the block so far, and 
$W_{j+1}$ is a given non-special vertex with probability at most $n^{\alpha-1}$.
So, we have
$$\p{S_1 \ge 2j+2 | S_1 \ge 2j} \ge 1 - 4j n^{\alpha-1}
\ge \exp ( - 8 j n^{\alpha-1} ) \:.
$$
Consequently,
$$
\p{S_1 > 2\ell}
\ge 
\prod_{j=1}^{\ell} \exp ( - 8 j n^{\alpha-1} )
= \exp \left ( -4 n^{\alpha-1} \ell(\ell+1) \right )
\ge 1/2
$$
by the choice of $\ell$, so \eqref{largeblocks} holds.

Observe that the block sizes $S_1,S_2,\dots$ are i.i.d.,
and each of them is at least $2\ell$ with probability at least $1/2$.
So we have
$$\p{ S_1 + \dots + S_k \le k \ell / 2} 
\le
\p{\operatorname{Bin}(k,1/2) \le k/4}
\le \exp(-k/16) = O(1/n^2) \:,
$$
where for the second inequality, we have used the Chernoff bound for binomials (see, e.g.,~\cite[Theorem~2.3(c)]{McD98}).
\end{proof}

\bibliographystyle{plain}
\bibliography{rumourspreadingN}

\begin{thebibliography}{10}

\bibitem{dregular_asynch}
H.~Amini, M.~Draief, and M.~Lelarge.
\newblock Flooding in weighted sparse random graphs.
\newblock {\em SIAM J. Discrete Math.}, 27(1):1--26, 2013.

\bibitem{BBCS05}
N.~Berger, C.~Borgs, J.T. Chayes, and A.~Saberi.
\newblock On the spread of viruses on the {I}nternet.
\newblock In {\em Proc.~16th Symp.\ Discrete Algorithms (SODA)}, pages
  301--310, 2005.

\bibitem{bol}
B.~Bollob{\'a}s and Y.~Kohayakawa.
\newblock On {R}ichardson's model on the hypercube.
\newblock In {\em Combinatorics, geometry and probability ({C}ambridge, 1993)},
  pages 129--137. Cambridge Univ. Press, Cambridge, 1997.

\bibitem{Boyd2006}
S.~Boyd, A.~Ghosh, B.~Prabhakar, and D.~Shah.
\newblock Randomized gossip algorithms.
\newblock {\em IEEE Transactions on Information Theory}, 52(6):2508--2530,
  2006.

\bibitem{DGH+87}
A.~Demers, D.~Greene, C.~Hauser, W.~Irish, J.~Larson, S.~Shenker, H.~Sturgis,
  D.~Swinehart, and D.~Terry.
\newblock Epidemic algorithms for replicated database maintenance.
\newblock In {\em Proc.~6th Symp.\ Principles of Distributed Computing (PODC)},
  pages 1--12, 1987.

\bibitem{DFF11}
B.~Doerr, M.~Fouz, and T.~Friedrich.
\newblock Social networks spread rumors in sublogarithmic time.
\newblock In {\em Proc.~43th Symp.\ Theory of Computing (STOC)}, pages 21--30,
  2011.

\bibitem{DFF12}
B.~Doerr, M.~Fouz, and T.~Friedrich.
\newblock Asynchronous rumor spreading in preferential attachment graphs.
\newblock In {\em Proc.~13th Scandinavian Workshop Algorithm Theory (SWAT)},
  pages 307--315, 2012.

\bibitem{experimental}
B.~Doerr, M.~Fouz, and T.~Friedrich.
\newblock Experimental analysis of rumor spreading in social networks.
\newblock In {\em Design and analysis of algorithms}, volume 7659 of {\em
  Lecture Notes in Comput. Sci.}, pages 159--173. Springer, Heidelberg, 2012.

\bibitem{richardson_model_survey}
R.~Durrett.
\newblock Stochastic growth models: recent results and open problems.
\newblock In {\em Mathematical approaches to problems in resource management
  and epidemiology ({I}thaca, {NY}, 1987)}, volume~81 of {\em Lecture Notes in
  Biomath.}, pages 308--312. Springer, Berlin, 1989.

\bibitem{robustness}
R.~Els{\"a}sser and T.~Sauerwald.
\newblock On the runtime and robustness of randomized broadcasting.
\newblock {\em Theoret. Comput. Sci.}, 410(36):3414--3427, 2009.

\bibitem{FPRU90}
U.~Feige, D.~Peleg, P.~Raghavan, and E.~Upfal.
\newblock Randomized broadcast in networks.
\newblock {\em Random Struct. Algorithms}, 1(4):447--460, 1990.

\bibitem{pemantle}
J.~A. Fill and R.~Pemantle.
\newblock Percolation, first-passage percolation and covering times for
  {R}ichardson's model on the {$n$}-cube.
\newblock {\em Ann. Appl. Probab.}, 3(2):593--629, 1993.

\bibitem{FP10}
N.~Fountoulakis and K.~Panagiotou.
\newblock Rumor spreading on random regular graphs and expanders.
\newblock In {\em Proc.~14th Intl.\ Workshop on Randomization and Comput.\
  (RANDOM)}, pages 560--573, 2010.

\bibitem{FPS12}
N.~Fountoulakis, K.~Panagiotou, and T.~Sauerwald.
\newblock Ultra-fast rumor spreading in social networks.
\newblock In {\em Proc.~23th Symp.\ Discrete Algorithms (SODA)}, pages
  1642--1660, 2012.

\bibitem{rgg}
T.~Friedrich, T.~Sauerwald, and A.~Stauffer.
\newblock Diameter and broadcast time of random geometric graphs in arbitrary
  dimensions.
\newblock {\em Algorithmica}, 67(1):65--88, 2013.

\bibitem{Gia11}
G.~Giakkoupis.
\newblock {Tight bounds for rumor spreading in graphs of a given conductance}.
\newblock In {\em 28th International Symposium on Theoretical Aspects of
  Computer Science (STACS 2011)}, volume~9, pages 57--68, 2011.

\bibitem{G13}
G.~Giakkoupis.
\newblock Tight bounds for rumor spreading with vertex expansion.
\newblock In {\em Proc.~25th Symp.\ Discrete Algorithms (SODA)}, pages
  801--815, 2014.

\bibitem{grimmett_book}
G.~R. Grimmett and D.~R. Stirzaker.
\newblock {\em Probability and random processes}.
\newblock Oxford University Press, New York, third edition, 2001.

\bibitem{Harchol-Balter1999}
M.~Harchol-Balter, F.~Thomson Leighton, and D.~Lewin.
\newblock Resource discovery in distributed networks.
\newblock In {\em Proc.~18th Symp.\ Principles of Distributed Computing
  (PODC)}, pages 229--237, 1999.

\bibitem{broadcasting_survey}
S.~M. Hedetniemi, S.~T. Hedetniemi, and A.~L. Liestman.
\newblock A survey of gossiping and broadcasting in communication networks.
\newblock {\em Networks}, 18(4):319--349, 1988.

\bibitem{fpp_book}
C.~D. Howard.
\newblock Models of first-passage percolation.
\newblock In {\em Probability on discrete structures}, volume 110 of {\em
  Encyclopaedia Math. Sci.}, pages 125--173. Springer, Berlin, 2004.

\bibitem{asynchronous_complete}
S.~Janson.
\newblock One, two and three times {$\log n/n$} for paths in a complete graph
  with random weights.
\newblock {\em Combin. Probab. Comput.}, 8(4):347--361, 1999.

\bibitem{KSSV00}
R.~Karp, C.~Schindelhauer, S.~Shenker, and B.~V\"ocking.
\newblock Randomized {R}umor {S}preading.
\newblock In {\em Proc.~41st Symp.\ Foundations of Computer Science (FOCS)},
  pages 565--574, 2000.

\bibitem{KDG03}
D.~Kempe, A.~Dobra, and J.~Gehrke.
\newblock Gossip-based computation of aggregate information.
\newblock In {\em Proc.~44th Symp.\ Foundations of Computer Science (FOCS)},
  pages 482--491, 2003.

\bibitem{McD98}
C.~McDiarmid.
\newblock Concentration.
\newblock In M.~Habib, C.~McDiarmid, J.~Ramirez-Alfonsin, and B.~Reed, editors,
  {\em Probabilistic Methods for Algorithmic Discrete Mathematics}, pages
  195--243. Springer-Verlag, 1998.

\bibitem{we_ktrees}
A.~Mehrabian and A.~Pourmiri.
\newblock Randomized rumor spreading in poorly connected small-world networks.
\newblock {\em arXiv}, 1410.8175 [cs.SI], 2014.
\newblock submitted (conference version in DISC 2014).

\bibitem{speidel}
K.~Panagiotou and L.~Speidel.
\newblock Asynchronous rumor spreading on random graphs.
\newblock In Leizhen Cai, Siu-Wing Cheng, and Tak-Wah Lam, editors, {\em
  Algorithms and Computation}, volume 8283 of {\em Lecture Notes in Computer
  Science}, pages 424--434. Springer Berlin Heidelberg, 2013.

\bibitem{mixing}
T.~Sauerwald.
\newblock On mixing and edge expansion properties in randomized broadcasting.
\newblock {\em Algorithmica}, 56(1):51--88, 2010.

\end{thebibliography}

\end{document}